\documentclass[manuscript,nonacm]{acmart}

\AtBeginDocument{%
	}

\usepackage{amsfonts,amssymb}
\usepackage[mathscr]{eucal}
\usepackage{graphicx}
\usepackage{color}
\usepackage{todonotes}
\usepackage{algorithm,algpseudocode}
\usepackage{xspace}
\usepackage{etoolbox}
\usepackage{xcolor}
\usepackage{multirow}
\usepackage{enumitem}
\usepackage[normalem]{ulem}
\usepackage{bm}
\usepackage{subcaption}
\usepackage{placeins}
\usepackage{balance}

\usepackage{mathtools}

\newtheorem{definition}{Definition}
\newtheorem{lemma}{Lemma}
\newtheorem{theorem}{Theorem}

\usepackage{tikz,pgfplots}
\usetikzlibrary{matrix, decorations, patterns, positioning, shapes, 3d, calc, intersections, arrows, fit, pgfplots.fillbetween}
\usepackage{tikz-3dplot}

\usepackage{cleveref}
\Crefname{definition}{Definition}{Definitions}
\crefname{definition}{Def.}{Defs.}
\Crefname{theorem}{Theorem}{Theorems}
\crefname{theorem}{Theorem}{Theorems}
\Crefname{lemma}{Lemma}{Lemmas}
\crefname{lemma}{Lemma}{Lemmas}
\Crefname{section}{Section}{Sections}
\crefname{section}{\S}{\S}
\crefname{algorithm}{Alg.}{Algs.}
\Crefname{algorithm}{Algorithm}{Algorithms}
\crefname{figure}{Fig.}{Figs.}
\Crefname{figure}{Figure}{Figures}
\crefname{table}{Tab.}{Tabs.}
\Crefname{table}{Table}{Tables}

\newcommand{\Real}{\mathbb{R}}
\newcommand{\V}[2][]{{\bm{#1\mathbf{\MakeLowercase{#2}}}}} 
\newcommand{\M}[2][]{{\bm{#1\mathbf{\MakeUppercase{#2}}}}} 
\newcommand{\T}[2][]{\boldsymbol{#1\mathscr{\MakeUppercase{#2}}}}

\newcommand{\sttsv}{STTSV\xspace}
\newcommand{\arraymul}{ternary multiplication}
\newcommand{\pluseq}{\mathrel{+}=}

\begin{document}

\title{Minimizing Communication for Parallel Symmetric Tensor Times Same Vector Computation}

\author{Hussam Al~Daas}
\orcid{0000-0001-9355-4042}
\affiliation{%
  \institution{STFC, Rutherford Appleton Laboratory}
  \city{Didcot}
  \state{Oxfordshire}
  \country{UK}
}
\email{hussam.al-daas@stfc.ac.uk}

\author{Grey Ballard}
\orcid{0000-0003-1557-8027}
\affiliation{%
  \institution{Wake Forest University}
  \city{Winston-Salem}
  \state{NC}
  \country{USA}}
\email{ballard@wfu.edu}

\author{Laura Grigori}
\orcid{0000-0002-5880-1076}
\affiliation{%
  \institution{EPFL and PSI}
  \city{Lausanne}
  \country{Switzerland}}
\email{laura.grigori@epfl.ch}

\author{Suraj Kumar}
\orcid{0009-0001-1449-0165}
\affiliation{%
  \institution{Inria Lyon}
  \city{Lyon}
  \country{France}
}
\email{suraj.kumar@inria.fr}

\author{Kathryn Rouse}
\orcid{0009-0001-4045-0423}
\affiliation{%
 \institution{Inmar Intelligence}
 \city{Winston-Salem}
 \state{NC}
 \country{USA}
}
\email{kathryn.rouse@inmar.com}

\author{Mathieu V\'{e}rit\'{e}}
\orcid{0000-0003-0581-2772}
\affiliation{%
 \institution{EPFL}
 \city{Lausanne}
 \country{Switzerland}
}
\email{mathieu.verite@epfl.ch}

\begin{abstract}
In this article, we focus on the parallel communication cost of multiplying the same vector along two modes of a $3$-dimensional symmetric tensor.
This is a key computation in the higher-order power method for determining eigenpairs of a $3$-dimensional symmetric tensor and in gradient-based methods for computing a symmetric CP decomposition.
We establish communication lower bounds that determine how much data movement is required to perform the specified computation in parallel.
The core idea of the proof relies on extending a key geometric inequality for $3$-dimensional symmetric computations.
We demonstrate that the communication lower bounds are tight by presenting an optimal algorithm where the data distribution is a natural extension of the triangle block partition scheme for symmetric matrices to 3-dimensional symmetric tensors.
\end{abstract}

\maketitle


\section{Introduction}
\label{sec:into}
Exploiting symmetry in matrix and tensor computations can save computation, memory, and communication, or movement of data.
As a symmetric matrix has approximately half the unique parameters of a non-symmetric matrix, many operations involving symmetric matrices require half the computation of their non-symmetric counterparts.
Because matrix computations are so fundamental, standard libraries such as implementations of the Basic Linear Algebra Subroutines (BLAS) commit to taking this constant factor advantage whenever possible.
More recent results have demonstrated that the communication costs for symmetric matrix computations can also be reduced with more complicated algorithms \cite{BELV22,ABGKR23,ABGKRV25,ABC+23}.

In the case of tensors, or multidimensional arrays with 3 or more modes, symmetry offers even more opportunity for cost savings.
For example, while a non-symmetric $d$-dimensional tensor of dimensions $n\times n \times \cdots \times n$ has $n^d$ entries, a symmetric tensor with the same dimensions has approximately $n^d/d!$ entries.
Here, we consider a symmetric tensor's entry to have the same value no matter how its indices are permuted, which is sometimes referred to as fully symmetric or supersymmetric.

Two fundamental problems involving a symmetric tensor are tensor eigenvalues, and Canonical Polyadic (CP) decomposition.
A $\mathbb{Z}$-eigenvalue of a 3-dimensional tensor $\T{A}$ is defined as a scalar $\lambda$ that satisfies the equation $\T{A} \times_2 \V{x} \times_3 \V{x} = \lambda \V{x}$ for some unit vector $\V{x}$ \cite{Lim2005,Qi2005}, where $\times_j$ denotes a tensor-times-vector operation in mode $j$.
A symmetric CP decomposition of a 3-dimensional tensor $\T{A}$ is specified by a matrix $\M{X}$ with $\ell$th column denoted $\V{x}_\ell$ such that $\T{A} \approx \sum_{\ell=1}^r \V{x}_\ell \circ \V{x}_{\ell} \circ \V{x}_{\ell}$, where $\circ$ denotes an outer product \cite{Kolda15}.
\Cref{alg:hopm} specifies a method for computing $\mathbb{Z}$-eigenvalues, and \Cref{alg:scp-grad} specifies the computation of the gradient of the objective function defined by the least squares error of the CP decomposition.
In both cases, the bottleneck computation is given by  $\T{A} \times_2 \V{x} \times_3 \V{x}$, which we refer to as the \emph{Symmetric-Tensor-Times-Same-Vector} (\sttsv) computation.
Algorithms for computing other types of eigenvalues and eigenvectors, including $\mathbb{H}$-eigenvalues, also rely on \sttsv{} \cite{SA+23}.
 
\begin{algorithm}
    \caption{\label{alg:hopm}Higher-Order Power Method}
    \begin{algorithmic}[1]
        \Require $\T{A}\in \Real^{n\times n\times n}$ is a symmetric 3-dimensional tensor. 
        \Ensure $\V{x}\in \Real^n$ and $\lambda \in \Real $ such that $\T{A} \times_2 \V{x} \times_3 \V{x} = \lambda \V{x}$.
        \State Initialize $\V{x}$ to a random unit vector
        \Repeat
        \State $\V{y} = \T{A} \times_2 \V{x} \times_3 \V{x} $\label{alg:hopm:bottleneck} \Comment{\sttsv}
        \State $\V{x} = \frac{\V{y}}{||\V{y}||}$
        \Until{convergence}
        \State $\lambda = \T{A} \times_1 \V{x} \times_2 \V{x} \times_3 \V{x}$
    \end{algorithmic}
\end{algorithm}

\begin{algorithm}
    \caption{\label{alg:scp-grad}Symmetric CP Gradient}
    \begin{algorithmic}[1]
        \Require $\T{A}\in \Real^{n\times n\times n}$ is a symmetric 3-dimensional tensor, $\M{X}\in \Real^{n\times r}$ is approximate CP factor
        \Ensure $\M{Y}\in \Real^{n\times r}$ is the gradient with respect to the function $f(\M{X})=\frac 16 \|\T{A} - \sum_\ell \V{x}_{\ell} \circ \V{x}_{\ell} \circ \V{x}_{\ell}\|^2$
        \State $\M{G} = (\M{X}^T\M{X}) \ast (\M{X}^T\M{X})$ \Comment{$\ast$ denotes the elementwise product}
        \For{$\ell = 1$ to $r$}
        \State $\V{y}_\ell = \T{A} \times_2 \V{x}_{\ell} \times_3 \V{x}_{\ell} $\label{alg:scp-grad:bottleneck} \Comment{\sttsv}
        \EndFor
        \State $\M{Y} = \M{X}\M{G}-\M{Y}$
    \end{algorithmic}
\end{algorithm}

In this work, we focus on \sttsv and its communication costs.
The main contributions of this paper are to
\begin{enumerate}
	\item derive a geometric inequality that bounds the maximum reuse of vector elements in a $3$-dimensional symmetric computation;
	\item establish communication lower bounds for the parallel load-balanced \sttsv computation;
	\item extend triangle block partitioning proposed by Beaumont et al.~\cite{BELV22} and Al Daas et al.~\cite{ABGKRV25} for symmetric matrices to tetrahedral block partitioning for 3-dimensional symmetric tensors; 
	\item present and analyze an algorithm whose communication cost exactly matches the leading term of the lower bound.
\end{enumerate}

To establish communication lower bounds, we employ a geometric inequality to determine the minimal amount of data that a processor requires to perform its share of computation and consider how much data the processor must send or receive.
In order to demonstrate that the communication lower bounds are tight, we provide an algorithm where the leading term of the communication cost exactly matches that of the lower bound. 
We extend the triangle block partition scheme proposed by Beaumont et al.~\cite{BELV22} and Al Daas et al.~\cite{ABGKRV25} for the lower triangle of symmetric matrices to a tetrahedral partition of the lower tetrahedron of 3-dimensional symmetric tensors. Our algorithm uses tetrahedral blocks to exploit the symmetry. To generate tetrahedral block partitions, we rely on Steiner systems that are collections of equally sized subsets of indices where three distinct indices appear in exactly one subset.
We use an owner-compute rule, assigning all computations involving elements in a block of the tensor to the processor that owns that block, ensuring that no tensor data needs to be communicated and only the input and output vectors need to be exchanged.

The rest of the paper is organized as follows. \Cref{sec:relatedwork} reviews previous work on communication lower bounds and optimal algorithms for some linear algebra computations. In \Cref{sec:notations}, we introduce our notations and preliminaries for the \sttsv computation. We present some key new results in \Cref{sec:fundamental}, which we use to establish communication lower bounds in \Cref{sec:bounds}. We explain tetrahedral block partitioning in \Cref{sec:tbc} and present a communication optimal algorithm for \sttsv in \Cref{sec:algo:3d}. Finally, we discuss conclusions and future perspectives in \Cref{sec:conclusion}.

\section{Related Work}
\label{sec:relatedwork}
There are numerous studies on establishing communication lower bounds. 
The first communication lower bounds were proposed by Hong and Kung~\cite{HK81} for various computations, including matrix multiplication, for a sequential machine using the red-blue pebble game. 
Aggarwal et al.~\cite{ACS90} extended their work for the parallel model and derived first memory-independent communication lower bounds for matrix multiplication. 
Irony et al.~\cite{IRONY:JPDC04} reproduced the parallel matrix multiplication bounds with the Loomis-Whitney geometric inequality~\cite{LW49}. 
Ballard et al.~\cite{Ballard:3NL} extended the use of Loomis-Whitney inequality to derive sequential and parallel bounds for all direct methods of linear algebra which can be expressed as three nested loops. 
Olivry et al.~\cite{OLPSR20} and Kwasniewski et al.~\cite{KB+21} analyzed computational graphs to automatically derive sequential lower bounds for several linear algebra computations.

Communication lower bounds have also been derived for computations involving tensors. 
In particular, Ballard et al.~\cite{BKR18,BR20,ABGKR24} gave sequential and parallel lower bounds for the Matricized-Tensor Times Khatri-Rao product (MTTKRP) and Multiple Tensor Times Matrix (Multi-TTM) computations using H\"{o}lder-Brascamp-Lieb inequalities \cite{BCCT10}.
Ziogas et al.~\cite{ZKBSH22} improved the leading constant for the sequential and memory dependent parallel lower bounds for MTTKRP using CDAG analysis.
Solomnik et al.~\cite{Solomonik:SISC:2021} gave communication lower bounds for symmetric tensor contractions.

While many of the automatically derived bounds for linear algebra computations were tight, the leading order constants were not tight for computations involving a symmetric matrix.
Beaumont et al.~\cite{BELV22} improved the constant for the sequential communication lower bound for SYRK using CDAG analysis by taking advantage of symmetry. 
They demonstrated the bound is tight by providing an algorithm that matches the lower bound. 
Their algorithm introduced triangle block partition of the symmetric matrix to allow maximum reuse of the elements of the non-symmetric matrix. 
Al Daas et al.~\cite{ABGKR23} derived a symmetric version of the Loomis-Whitney inequality which they used to obtain memory-independent parallel communication lower bounds for SYRK. 
They adapted the triangle block partition given by Beaumont et al.~\cite{BELV22} for the parallel case and provided communication optimal algorithms for SYRK.
Agullo et al.~\cite{ABC+23} were able to improve the performance of distributed SYMM implementations by applying the triangle block partition to improve the arithmetic intensity over standard implementations.
Al Daas et al.~\cite{ABGKRV25} established sequential and parallel communication lower bounds for three nested loop computations involving a symmetric matrix.
They gave algorithms for SYRK, SYR2K and SYMM all of which achieve the communication lower bounds, using new families of triangle block partitions derived from finite geometries. 
We follow their approach and use finite geometries to obtain partitions for 3-dimensional symmetric tensors.

\section{Notations and Preliminaries}
\label{sec:notations}

We will use the following notation throughout the paper.
$\T{A}$ represents a $3$-dimensional symmetric tensor of dimensions $n\times n \times n$, while $\V{x}$ and $\V{y}$ are vectors of length $n$. 
As $\T{A}$ is a symmetric tensor, $a_{ijk}$ is same for all permutations of $i$, $j$ and $k$, i.e, $a_{ijk} = a_{ikj} = a_{jik} = a_{jki} = a_{kij} = a_{kji}$.   

The tensor-times-vector operation in mode $j$ of $\T{A}$ and $\V{x}$ is denoted as $\T{A}\times_j \V{x}$ and it can be computed element-wise as $(\T{A}\times_j \V{x})_{ik} = \sum_j a_{ijk}\cdot x_j$. 
Similarly, $\T{A}\times_j \V{x} \times_k \V{x}$ can be computed element-wise as $(\T{A}\times_j \V{x} \times_k \V{x})_{i} = \sum_{j,k} a_{ijk}\cdot x_j \cdot x_k$. 
We refer to the operation $a_{ijk}\cdot x_j \cdot x_k$ as \emph{\arraymul}.

We present two algorithms for performing the \sttsv computation.
The first, \Cref{alg:sttvv:naive}, is straightforward but does not take symmetry into account. 
It performs the computation by performing all $n^3$ ternary multiplications.

\begin{algorithm}[htb]
	\caption{\label{alg:sttvv:naive}Pseudocode of \sttsv}
	\begin{algorithmic}
		\Require $\T{A}\in \Real^{n\times n\times n}$ is a symmetric tensor and $\V{x}\in \Real^n$ is a vector.
		\Ensure $\V{y}\in \Real^n$ such that  $\V{y} = \T{A} \times_2 \V{x} \times_3 \V{x} $.
		\State Initialize vector $\V{y}$ to $0$
		\For{$i=1 \text{ to } n$}
		\For{$j=1 \text{ to } n$}
		\For{$k=1 \text{ to } n$}
		\State $y_i \pluseq a_{ijk}\cdot x_j \cdot x_k$
		\EndFor
		\EndFor
		\EndFor
	\end{algorithmic}
\end{algorithm}

Our second algorithm, \Cref{alg:sttvv:symmetry}, exploits the property that every element of a symmetric tensor remains constant under any permutation of its indices.
It only works with the lower tetrahedral portion of $\T{A}$, $a_{ijk}$ where $i\geq j\geq k$, performing all possible operations involving $a_{ijk}$ within each inner loop iteration. 
The total number of points in the iteration space is $n(n+1)(n+2)/6$ of which $n(n-1)(n-2)/6$ correspond to computations with the strict lower tetrahedral portion of $\T{A}$, $a_{ijk}$ where $i> j> k$. 

\begin{algorithm}[htb]
	\caption{\label{alg:sttvv:symmetry}\sttsv exploiting symmetric structure}
	\begin{algorithmic}
		\Require $\T{A}\in \Real^{n\times n\times n}$ is symmetric and $\V{x}\in \Real^n$
		\Ensure $\V{y}\in \Real^n$ such that  $\V{y} = \T{A} \times_2 \V{x} \times_3 \V{x} $
		\State Initialize vector $\V{y}$ to $0$
		\For{$i=1 \text{ to } n$}
		\For{$j=1 \text{ to } i$}
		\For{$k=1 \text{ to } j$}
		\If{$i\neq j$ and $j\neq k$}
		\State $y_i \pluseq 2 a_{ijk}\cdot x_j\cdot x_k$
		\State $y_j \pluseq 2a_{ijk}\cdot x_i\cdot x_k$
		\State $y_k\pluseq 2a_{ijk}\cdot x_i\cdot x_j$
		\ElsIf{$i==j$ and $j\neq k$}
		\State $y_i \pluseq 2a_{ijk}\cdot x_j \cdot x_k$
		\State $y_k \pluseq a_{ijk}\cdot x_i \cdot x_j$
		\ElsIf{$i\neq j$ and $j==k$}
		\State $y_i \pluseq a_{ijk}\cdot x_j \cdot x_k$
		\State $y_j \pluseq 2a_{ijk}\cdot x_i\cdot x_k$
		\Else
		\State $y_i \pluseq a_{ijk}\cdot x_j\cdot x_k$
		\EndIf
		\EndFor
		\EndFor
		\EndFor
	\end{algorithmic}
\end{algorithm}

\Cref{alg:sttvv:symmetry} performs $n^2(n+1)/2$ ternary multiplications, approximately half the number of those in \Cref{alg:sttvv:naive}.

\begin{definition}
	\label{def:atomic-3-array}
	A parallel atomic \sttsv algorithm computes each \arraymul on one processor.
\end{definition}

We consider only parallel atomic algorithms for \sttsv, that perform $n^2(n+1)/2$ \arraymul.
Computing a single \arraymul on one processor requires that all $3$ inputs are accessed by that processor to compute the single output value. 
This requirement is necessary for our communication lower bounds to hold. 

\subsection{Parallel Computation Model}
\label{sec:model}

We use the $\alpha-\beta-\gamma$ or MPI model \cite{Thakur:CollectiveCommunications:2005,Chan:CollectiveCommunications:2007} of parallel computation.
In this model, the computation is distributed over $P$ processors each of which has its own local memory and can operate on data in its local memory. 
In order to access data from other processors, a processor must communicate with that processor via a fully connected network with bidirectional links.
Each processor can send and receive at most one message at the same time. 
Hence, communication refers to send and receive operations that transfers data from local memory to the network and vice-versa. 
The communication cost has two components, the latency cost or number of messages, and bandwidth cost or volume of communication. 
As the bandwidth cost tends to dominate the latency cost when messages are large, we focus on bandwidth cost and refer to it as communication cost throughout the paper.
We focus on memory-independent analysis, so we assume that each processor has sufficiently large local memory.

\section{Geometric Results}
\label{sec:fundamental}
To minimize communication in parallel computations, we minimize the amount of data required to perform each processor's local computations.
We relate data to computations by relating the set of points in $\mathbb{Z}^3$ which correspond to elements in the iteration space to projections onto lower dimensional spaces which correspond to the arrays of data.
Our first lemma relates the number of points in a three dimensional lattice to the sizes of its projections onto one dimensional subspaces, relating points in a three dimensional iteration space to vector data required for the computation.

We obtain the following lemma from \cite[Lemma 4.1]{BKR18} by fixing the number of dimensions, it can also be proved by applying Loomis-Whitney inequality~\cite{LW49} twice. 
\begin{lemma}
	\label{lem:basicHBL}
	Let $V$ be a finite set of points in $\mathbb{Z}^3$. Let $\phi_i(V)$ be the projection of $V$ on the $i$-axis, i.e., all points $i$ such that there exists a $(j,k)$ so that $(i,j,k) \in V$. Define $\phi_j(V)$ and $\phi_k(V)$ similarly. 
	Then 
	$$|V| \leq |\phi_i(V)| \cdot |\phi_j(V)| \cdot |\phi_k(V)|.$$
\end{lemma}

We now extend \Cref{lem:basicHBL} for a symmetric set of points.
\begin{lemma}
	\label{lem:symmHBL}
	Let $V$ be a finite set of points contained in $\{(i, j, k)\in\mathbb{Z}^3 \, | \, i>j>k\}$.  Let $\phi_i(V)$ be the projection of $V$ on the $i$-axis, i.e., all points $i$ such that there exists a $(j,k)$ so that $(i,j,k) \in V$. Define $\phi_j(V)$ and $\phi_k(V)$ similarly. 
	Then 
		$$6|V| \leq |\phi_i(V) \cup \phi_j(V) \cup \phi_k(V)|^3.$$
\end{lemma}
\begin{proof}
    Consider the set $$\tilde{V} = \{(i,j,k), (i,k,j), (j, i, k), (j,k,i), (k, i, j), (k, j, i) \mid (i, j, k)\in V \}$$ and its projections. 
    To begin, we will show that $|\tilde{V}| = 6|V|$. 
    For each $(i, j, k)\in V $, there are $6$ corresponding elements in $\tilde{V}$. 
    Therefore $|\tilde{V}|\leq 6|V|$. 
    If $(i, j, k), (i', j', k')\in V$ and $(i, j, k) \neq (i', j', k')$, then $\{\text{all permutations of } (i, j, k)\} \cap \{\text{all permutations of } (i', j', k')\} = \emptyset$. 
    Therefore $|\tilde{V}|\geq 6|V|$. Thus we have $|\tilde{V}| = 6|V|$.
	
	Now we will show that $\phi_i(\tilde{V}) = \phi_i(V) \cup \phi_j(V) \cup \phi_k(V)$. If $i \in \phi_i(\tilde{V})$, then there exists a $(j, k)$ such that $(i, j, k)\in V$ or $(j, i, k)\in V$ or $(j, k, i) \in V$. 
	Therefore $\phi_i(\tilde{V})\subseteq \phi_i(V)\cup\phi_j(V) \cup \phi_k(V)$. 
	If $i \in \phi_i(V) \cup \phi_j(V) \cup \phi_k(V)$, then there exists a $(j, k)$ such that $(i, j, k)\in V$ or $(j, i, k)\in V$ or $(j, k, i) \in V$, thus $(i, j, k)\in\tilde{V}$. 
	Therefore $\phi_i(V)\cup\phi_j(V) \cup \phi_k(V) \subseteq \phi_i(\tilde{V})$. 
	Thus we have $\phi_i(\tilde{V}) = \phi_i(V) \cup \phi_j(V) \cup \phi_k(V)$. 
	The same set of arguments shows that $\phi_j(\tilde{V}) = \phi_k(\tilde{V}) = \phi_i(V) \cup \phi_j(V) \cup \phi_k(V)$.
	
	By \Cref{lem:basicHBL}, we know that $|\tilde{V}| \leq |\phi_i(\tilde{V})||\phi_j(\tilde{V})||\phi_k(\tilde{V})|$. 
	Substituting the above results in this inequality yields $6|V| \leq |\phi_i(V) \cup \phi_j(V) \cup \phi_k(V)|^3$.
\end{proof}

\Cref{lem:symmHBL} indicates that when $\phi_i(V) = \phi_j(V) = \phi_k(V)$, these projections can be utilized in at most $\frac{|\phi_i(V)|^3}{6}$ points of a $3$-dimensional symmetric iteration space.

\section{Memory-Independent Communication Bounds}
\label{sec:bounds}
Our lower bound arguments minimize the total amount of data required to perform a processor's assigned computation.
Ultimately our lower bound relies on a constrained minimization problem which we state abstractly in \cref{lem:opt}.
The objective function of the minimization problem gives the number of elements which must be accessed to perform the required computation.
The first constraint comes from the assumption that the processor performs $1/P$th of the ternary multiplications associated with the strict lower tetrahedron of the symmetric tensor.
The second constraint comes from \cref{lem:symmHBL}.

We can prove the following lemma by taking minimum values of both variables.
\begin{lemma}
	\label{lem:opt}
	Consider the following optimization problem:
	$$\min x_1+ 2x_2$$
	such that 
	$$\frac{n(n-1)(n-2)}{6P} \leq x_1,$$
	$$\frac{n(n-1)(n-2)}{P} \leq x_2^3,$$
	where $P\geq 1 $, and $n$ is a positive integer. The optimal solution occurs at $\left(\frac{n(n-1)(n-2)}{6P}, \left(\frac{n(n-1)(n-2)}{P}\right)^{1/3}\right)$.
\end{lemma}

\subsection{Communication Lower Bounds}
\label{sec:bounds:mainresults}
We now apply the solution of \cref{lem:opt} to derive the communication bounds for the \sttsv computation, presented as \Cref{thm:memindeplb}. 
We consider a single processor that performs the computations associated with $1/P$th of iteration points and has access to $1/P$th of the data. 
To obtain the lower bound on communication, we subtract the quantity of data the processor stores at the start and end of the computation from the required amount of data to perform its computations.

\begin{theorem}
	\label{thm:memindeplb}
	Consider the \sttsv computation, $\V{y}=\T{A}\times_2 \V{x}\times_3 \V{x}$, where $\V{x}$ and $\V{y}$ are vectors of length $n$ and $\T{A}$ has dimensions $n\times n\times n$. 
	Suppose a parallel atomic algorithm using $P$ processors begins with one copy of $\V{x}$ and the strict lower tetrahedron of $\T{A}$, and ends with one copy of $\V{y}$. 
	If each processor performs the ternary multiplications associated with $1/P$th of the points in the iteration space involving the strict lower tetrahedron of $\T{A}$, then at least one processor must communicate at least $ 2\left(\frac{n(n-1)(n-2)}{P}\right)^{1/3} - 2 \frac{n}{P}$ elements.
\end{theorem}

\begin{proof}

	To begin we note that there are $\frac{n(n-1)(n-2)}{6}$ elements in the strict lower tetrahedron of $\T{A}$ and $n$ elements each in $\V{x}$ and $\V{y}$.
	There must be a processor that owns at most $1/P$th of these elements as otherwise input and output elements need to be replicated contradicting our assumptions.

	We determine how many elements this processor must access to perform the ternary multiplications of the assigned iteration points. Let $F$ be the set of points $(i,j,k)$ associated with the strict lower tetrahedron of $\T{A}$ assigned to this processor.
	Then $F\subseteq \{(i,j,k)\in\mathbb{Z}^3|i>j>k\}$ and $|F| \geq \frac{n(n-1)(n-2)}{6P}$
	by our assumption that this processor performs $\frac{n(n-1)(n-2)}{6P}$ iteration points.
	
	By \Cref{lem:symmHBL}, we know that 
	$$|\phi_i(F) \cup \phi_j(F) \cup \phi_k(F)|^3 \geq 6|F| \geq \frac{n(n-1)(n-2)}{P}\text.$$ 
	
	For each point, in $F$, exactly one element of $\T{A}$ is accessed, thus the total number of accessed elements of $\T{A}$ is $|F|$. 
	The union of projections of $F$, $\phi_i(F) \cup \phi_j(F) \cup \phi_k(F)$, gives the indices of the elements of each vector $\V{x}$ and $\V{y}$ accessed or modified by the processor. 
	To minimize the communication we want to minimize the number of elements accessed by this processor.
	Thus we want to minimize $|F| + 2|\phi_i(F) \cup \phi_j(F) \cup \phi_k(F)|$ subject to the above constraints. 
	By \Cref{lem:opt}, the minimum number of elements that must be accessed by this processor is $\frac{n(n-1)(n-2)}{6P} + 2\left(\frac{n(n-1)(n-2)}{P}\right)^{1/3}$. 

	As the processor owns at most $\frac{n(n-1)(n-2)}{6P} + \frac{2n}{P}$ elements, the number of elements communicated must be at least
	$$\frac{n(n-1)(n-2)}{6P} + 2\left(\frac{n(n-1)(n-2)}{P}\right)^{1/3} - \left(\frac{n(n-1)(n-1)}{6P}+\frac{2n}{P} \right)= 2\left(\frac{n(n-1)(n-2)}{P}\right)^{1/3} - 2 \frac{n}{P}\text.$$
\end{proof}

\section{Tetrahedral Block Partitioning}
\label{sec:tbc}

Communication optimal algorithms for several symmetric matrix computations have been proposed recently~\cite{BELV22,ABGKR23,ABGKRV25,ABC+23}. 
These algorithms improve the communication complexity for computations involving symmetric matrices by partitioning the lower triangle of a symmetric matrix into triangle blocks which allows all computations involving an element of the symmetric matrix to be computed together. 
A triangle block is the set of pairs corresponding to the strict lower triangle of the product of a set of indices with itself. 
For example, given a set of indices $\{1,4,6,8\}$, the associated triangle block is the set of pairs $\{(4,1), (6,1), (8,1), (6,4), (8,4), (8,6)\}$. 
By using a triangle block partition of the symmetric matrix, a processor can perform operations with $\mathcal{O}(r^2)$ elements of the symmetric matrix by using only $r$ elements from the same column of a non-symmetric matrix. 

We generalize the definition of triangle blocks to $3$-dimensions so that we can partition the lower tetrahedron of a symmetric $3$-dimensional tensor into tetrahedral blocks. 
Let $R$ be a subset of indices of the symmetric $3$-dimensional tensor, $R\subseteq \{1,\ldots,n\}$, then the tetrahedral block defined by $R$ is $TB_3(R)=\{(i,j,k) \; | \; i,j,k\in R, i > j >k\}$. 
For example, $TB_3(\{1,4,6,8\}) = \{(6,4,1), (8,4,1), (8,6,1), (8,6,4)\}$. 
In order to partition the lower tetrahedron of a symmetric $3$-dimensional tensor into tetrahedral blocks, we rely on Steiner systems which are well studied combinatorial objects.

\begin{definition}{\cite[Section 3]{Gowers2017}}
	A Steiner $(n, r, s)$-system is a collection $\Sigma$ of subsets of size $r$ from the set $S=\{1,\ldots,n\}$ such that every subset of $S$ of size $s$ is contained in exactly one set from $\Sigma$.
\end{definition}

While general results on the existence of Steiner systems are due to Keevash~\cite{Keevash2014}, existence results for $s=3$ are due to Wilson~\cite{Wilson1975}.

\begin{theorem}{\cite[Corollary A]{Wilson1975}}
	\label{thm:divConditions}
	Given a positive integer $r$, there exists an $n_0\in\mathbb{Z}$ such that Steiner $(n,r,3)$ systems exist for all $n \geq n_0$ for which $r-2$ divides $n-2$, $(r-1)(r-2)$ divides $(n-1)(n-2)$ and $r(r-1)(r-2)$ divides $n(n-1)(n-2)$.
\end{theorem}

The following lemmas, which follow directly from Theorem 3.3 of~\cite{crcCombinatorialDesigns1996}, allow us to count the number of subsets in a Steiner $(n,r,3)$ system which share two or 1 element subsets of the original set $S=\{1,\ldots,n\}$.

\begin{lemma}
\label{lem:countSubsetsContainPair}
If $\Sigma$ is a Steiner $(n,r,3)$ system and $X$ any 2 element subset of $S$, then the number of subsets in $\Sigma$ containing $X$ is $\frac{n-2}{r-2}$.
\end{lemma}

\begin{lemma}
	\label{lem:countSubsetsContainElement}
If $\Sigma$ is a Steiner $(n,r,3)$ system then any element of $S$ appears in exactly $\frac{(n-1)(n-2)}{(r-1)(r-2)}$ subsets.
\end{lemma}

A well known infinite family of Steiner $(n,r,3)$ systems can be constructed from finite spherical geometries. 
\begin{theorem}{\cite[Example 3.23]{crcCombinatorialDesigns1996}}
\label{thm:sphericalDesigns}
Let $q$ be a prime power and $\alpha > 0$ an integer. 
Then $G=\rm{PGL}_2(q^\alpha)$ acts sharply 3-transitively on $\mathbb{F}_{q^\alpha}\cup\{\infty\}$.
If $S\subseteq \mathbb{F}_{q^\alpha}\cup\{\infty\}$ is the natural inclusion of $\mathbb{F}_q\cup\{\infty\}$, then the orbit of $S$ under the action of $G$ is a Steiner $(q^\alpha+1,q+1,3)$ system.
\end{theorem}
To construct these systems, one can compute the images of $\mathbb{F}_q\cup\{\infty\}$ under one element of each coset in $PGL_2(q^\alpha)/PGL_2(q)$.

A Steiner $(q^\alpha + 1,q+1,3)$ system has $\frac{(q^\alpha+1)q^\alpha(q^\alpha-1)}{(q+1)q(q-1)}$ subsets. 
A tetrahedral block can be constructed from each subset.
We will use tetrahedral block constructions from Steiner $(q^2+1,q+1,3)$ systems to simplify the analysis of our algorithm.
In this case, there are $|\Sigma|= \frac{(q^2+1)q^2(q^2-1)}{(q+1)q(q-1)} = q(q^2+1)$ blocks, any index appears in $\frac{q^2(q^2-1)}{q(q-1)}= q(q+1)$ blocks by \Cref{lem:countSubsetsContainElement}, and two distinct indices together appear in $\frac{q^2-1}{q-1} = q+1$ blocks by \Cref{lem:countSubsetsContainPair}.

\subsection{Data distributions and assignment of computations}
\label{sec:dataDistributions}
As the communication lower bounds of \Cref{sec:bounds} do not contain any terms from communicating the symmetric tensor, we design algorithms that only communicate entries from the vectors.
To do this, each processor begins the computation with a tetrahedral block of the symmetric tensor and performs all computations that involve elements of the block.
Thus the number of processors $P$ must match the number of tetrahedral blocks produced by our blocking scheme, and we assume that $P=q(q^2+1)$ for some prime power $q$.

In order to partition the lower tetrahedron of a tensor with dimension $n$ using a Steiner $(q^2+1,q+1,3)$ system we follow Al Daas et al.~\cite{ABGKR23} and use a tetrahedral block of blocks.
For the given $n$ and $P=q(q^2+1)$, we divide the indices  $\{1,\ldots,n\}$ into contiguous parts, each of length $b=n/(q^2+1)$.  
Then we work with $b\times b\times b$ blocks of the symmetric tensor, and $b$ contiguous elements of the vector.
If $q^2+1$ does not divide $n$, we can pad our tensor to dimension $n'$, where $n'$ is the smallest multiple of $q^2+1$ that is at least $n$, satisfying the divisibility condition and set $b=n'/(q^2+1)$.

To simplify our presentation, we name different types of blocks from the symmetric tensor.
As we only consider the blocks in the lower tetrahedron of $\T{A}$, all block indices satisfy $i\geq j\geq k$.

\begin{itemize}
	\item \emph{off-diagonal block}: a block with index $(i, j, k)$ is called \emph{off-diagonal} if $i > j > k$.
    \item \emph{non-central diagonal block}: a block is called \emph{non-central diagonal} if exactly two of its indices are equal.
    	\item \emph{central diagonal block}: a block is called \emph{central diagonal} if all its indices are equal.
\end{itemize}

There are $(q^2+1)(q^2+2)(q^2+3)/6$ blocks in the lower tetrahedron of the symmetric tensor where $(q^2+1)q^2(q^2-1)/6$ blocks are off diagonal, $q^2(q^2+1)$ blocks are non-central diagonal and $q^2+1$ blocks are central diagonal.
The Steiner system partitions the off diagonal blocks, but we need to assign the diagonal blocks to processors in a way that is compatible with the assignment of off diagonal blocks, specifically in a way that their computations do not require any vector elements that are not already required by the processor.

Let $\{R_p\}$ be the set of subsets of $\{1,\ldots, q^2+1\}$ defined by the Steiner $(q^2+1,q+1,3)$ system, so $|R_p| = q+1$.
We let $N_p$ and $D_p$ represent a set of non-central diagonal blocks and a set of central diagonal blocks, respectively, assigned to processor $p$. 
\Cref{tab:TBP-10-30} shows the assignment of all lower tetrahedral blocks of a symmetric tensor from the Steiner $(10,4,3)$ system. 
We will work with this example throughout the text. 

We first discuss the data distribution of the off-diagonal blocks of the symmetric tensor, followed by the data distribution of both vectors, and then the data distribution of the non-central and central diagonal blocks of the symmetric tensor.

\begin{table}[t]
	\begin{tabular}{|c|c|c|c|c|}
		\hline
		$p$ & $R_p$ & $N_p$ & $D_p$ \\
		\hline
		1 & \{1,2,3,7\} & \{(2,2,1), (2,1,1), (7,2,2)\} & \{(1,1,1)\} \\
		2 & \{1,2,4,5\} & \{(4,4,1), (4,1,1), (5,1,1)\} & \{(2,2,2)\} \\
		3 & \{1,2,6,10\} & \{(6,6,1),(10,10,2), (6,1,1)\} & \{(6,6,6)\} \\
		4 & \{1,2,8,9\} & \{(8,8,1),(9,9,8), (8,1,1)\} & \{(8,8,8)\} \\
		5 & \{1,3,4,10\} & \{(10,10,1),(10,10,3), (10,1,1)\} & \{(3,3,3)\} \\
		6 & \{1,3,5,8\} & \{(3,3,1),(8,8,5), (3,1,1)\} & \{(5,5,5)\} \\
		7 & \{1,3,6,9\} & \{(9,9,1),(9,9,3), (9,1,1)\} & \{(9,9,9)\} \\
		8 & \{1,4,6,8\} & \{(6,6,4),(8,8,6), (6,4,4)\} & \{(4,4,4)\} \\
		9 & \{1,4,7,9\} & \{(7,7,1), (9,9,4), (7,1,1))\} & \{(7,7,7)\} \\
		10 & \{1,5,6,7\} & \{(5,5,1), (7,7,6), (7,6,6)\} & \{\} \\
		11 & \{1,5,9,10\} & \{(9,9,5),(10,10,9), (9,5,5)\} & \{(10,10,10)\}  \\
		12 & \{1,7,8,10\} & \{(8,8,7),(10,10,8), (10,8,8)\} & \{\} \\
		13 & \{2,3,4,8\} & \{(3,3,2), (3,2,2),(4,2,2)\} & \{\}  \\
		14 & \{2,3,5,6\} & \{(5,5,2), (5,2,2),(6,5,5)\} & \{\}  \\
		15 & \{2,3,9,10\} & \{(9,9,2), (9,2,2),(10,2,2)\} & \{\}  \\
		16 & \{2,4,6,9\} & \{(4,4,2), (9,9,6), (9,6,6)\} & \{\}  \\
		17 & \{2,4,7,10\} & \{(7,7,2),(10,10,4), (10,4,4)\} & \{\}  \\
		18 & \{2,5,7,9\} & \{(7,7,5),(9,9,7), (7,5,5)\} & \{\} \\
		19 & \{2,5,8,10\} & \{(8,8,2), (8,2,2),(10,5,5)\} & \{\}  \\
		20 & \{2,6,7,8\} & \{(6,6,2), (6,2,2),(8,6,6)\} & \{\}  \\
		21 & \{3,4,5,9\} & \{(4,4,3), (4,3,3),(9,4,4)\} & \{\}  \\
		22 & \{3,4,6,7\} & \{(6,6,3), (6,3,3),(7,3,3)\} & \{\}  \\
		23 & \{3,5,7,10\} & \{(5,5,3), (5,3,3),(10,3,3)\} & \{\}  \\
		24 & \{3,6,8,10\} & \{(8,8,3),(10,10,6), (8,3,3)\} & \{\}  \\
		25 & \{3,7,8,9\} & \{(7,7,3), (9,7,7),(9,3,3)\} & \{\}  \\
		26 & \{4,5,6,10\} & \{(5,5,4), (5,4,4),(10,10,5)\} & \{\}  \\
		27 & \{4,5,7,8\} & \{(7,7,4), (7,4,4),(8,7,7)\} & \{\}  \\
		28 & \{4,8,9,10\} & \{(8,8,4), (8,4,4),(10,9,9)\} & \{\}  \\
		29 & \{5,6,8,9\} & \{(6,6,5), (8,5,5) (9,8,8)\} & \{\} \\
		30 & \{6,7,9,10\} & \{(10,6,6),(10,10,7), (10,7,7)\} & \{\}  \\
		\hline
	\end{tabular}
	\caption{
		Processor sets of tetrahedral block partition for $m=10$ and $P=30$. 
		Processor $p$ owns all $TB_3(R_p) \cup N_p \cup D_p$ blocks of the symmetric tensor and performs all required computations with them.} 
	\label{tab:TBP-10-30}
\end{table}

\begin{table}[t]
	\begin{tabular}{|c|c|}
		\hline
		$i$ & $Q_i$\\
		\hline
		 1 & \{1,2,3,4,5,6,7,8,9,10,11,12\}\\
		 2 & \{1,2,3,4,13,14,15,16,17,18,19,20\}\\
		 3 & \{1,5,6,7,13,14,15,21,22,23,24,25\}\\
		 4 & \{2,5,8,9,13,16,17,21,22,26,27,28\}\\
		 5 & \{2,6,10,11,14,18,19,21,23,26,27,29\}\\
		6 & \{3,7,8,10,14,16,20,22,24,26,29,30\}\\
		 7 & \{1,9,10,12,17,18,20,22,23,25,27,30\}\\
		 8 & \{4,6,8,12,13,19,20,24,25,27,28,29\}\\
		 9 & \{4,7,9,11,15,16,18,21,25,28,29,30\}\\
		10 & \{3,5,11,12,15,17,19,23,24,26,28,30\}\\
		\hline
	\end{tabular}
	\caption{
		Row block sets of tetrahedral block partition for $m=10$ and $P=30$. 
		Row block $i$ of a vector is evenly distributed  among all processors of $Q_i$.
	} 
	\label{tab:PRS-10-30}
\end{table}

\subsubsection{Data distribution of off-diagonal blocks} 
\label{sec:assignment:offdiagonals}
The blocks corresponding to indices of $TB_3(R_p)$ are assigned to processor $p$. 
As $|R_p| = q+1$, processor $p$ has all $(q+1)q(q-1)/6$ off-diagonal blocks corresponding to $TB_3(R_p)$ at the beginning of the computation, and will perform all computations that require entries from these blocks. 
For example, $R_3= \{1,2,6,10\}$ in \Cref{tab:TBP-10-30} indicates that processor $3$ has blocks $(6,2,1), (10,2,1), (10,6,1)$ and $(10,6,2)$ of the symmetric tensor and will perform all computation involving these blocks.

\subsubsection{Data distribution of both vectors} 
We distribute row blocks of vectors $\V{x}$ and $\V{y}$ to be compatible with the data required by each processor. 
By~\Cref{lem:countSubsetsContainElement}, each row block of the vectors is required on $q(q+1)$ processors, so we distribute the row block equally on these processors. 
We assume that $b=\frac{n}{q^2+1}\geq q(q+1)$ to ensure that we can divide each row block across all the processors that require it.
\Cref{tab:PRS-10-30} shows which processors need a particular row block from each vector. 
For example, row block $3$ from both vectors $\V{x}$ and $\V{y}$ is required by processors $1,5,6,7,13,14,15,21,22,23,24$ and $25$. 
We use $Q_i$ to denote the set of processors that requires the $i$th block of each vector. 
We obtain $\{Q_i\}$ by analyzing the sets $R_p$. 

As each processor works with $q+1$ row blocks of each vector $\V{x}$ and $\V{y}$, each processor has $(q+1)\frac{b}{q(q+1)} = \frac{n}{P}$ elements of $\V{x}$ at the beginning of the computation and the same number of elements of $\V{y}$ at the end of the computation.

\subsubsection{Data distribution of non-central and central diagonal blocks}
\label{sec:assignment:diagonals}
We first discuss the distribution of $q^2(q^2+1)$ non-central diagonal blocks of the symmetric tensor.
We want to assign blocks to a processor in a compatible way with the off-diagonal blocks.
Specifically we want to assign the non-central diagonal blocks in such a way that performing their associated computations only requires data from $\V{x}$ and $\V{y}$ that is already required to perform the computations associated with the off-diagonal blocks assigned to the processor.
For example, processors $1,2,3$ and $4$ are the potential candidates for the block $(2,2,1)$ based on the sets $R_p$ defined in~\Cref{tab:TBP-10-30} as these processors already require the $1$st and $2$nd row blocks of both vectors. 
We assign $q^2(q^2+1)/P = q$ blocks to each processor.

Assigning disjoint sets of $q$ non-central diagonal blocks to each processor is equivalent to finding $q$ disjoint matchings in the bipartite graph  $G(X, Y, E)$ with $X =\{1,\ldots,P\}$ corresponding to processors, $Y = \{(a,a,b): q^2+1\geq a > b\geq 1\}\cup\{(a,b,b):q^2+1\geq a > b \geq 1\}$ corresponding to the non-central diagonal blocks, and $E = \{(x,y)\in X\times Y: y= (a,b,b) \text{ or } y=(a,a,b)\text{ and } a,b \in R_x\}$.
To show that such an assignment is possible, we use the following results.

\begin{theorem}[Hall's Marriage Theorem, {\cite[Theorem 2.1.2]{Diestel2017}}]
	\label{thm:hall:graphformulation}
	Let $G(X,Y,E)$ be a finite bipartite graph with bipartite sets $X$ and $Y$ and edge set $E$. There exists a matching with a set of disjoint edges that covers every vertex in $X$ iff $|W| \leq |N_G(W)|$ for every subset $W\subseteq X$. Here $N_G(W)$ denotes the set of all vertices in $Y$ that are adjacent to at least one vertex of $W$.
\end{theorem}

We obtain the following corollary by applying \Cref{thm:hall:graphformulation} on an extended graph.

\begin{corollary}
	\label{cor:hall:dgraphformulation}
	Let $G(X,Y,E)$ be a finite bipartite graph with bipartite sets $X$ and $Y$ and edge set $E$. There exist $d$ disjoint matchings if $d|W| \leq |N_G(W)|$ for every subset $W\subseteq X$. Here $N_G(W)$ denotes the set of all vertices in $Y$ that are adjacent to at least one vertex of $W$. Every matching consists of a set of disjoint edges that covers each vertex in $X$ and exactly $|X|$ unique vertices in $Y$. 
\end{corollary}
\begin{proof}
	We transform $G$ into a new bipartite graph $G_d$ by replacing each vertex in $X$ with $d$ copies of itself, and connect each copy as in the original graph. We apply Hall's theorem to the transformed graph. If $|W| \leq d|N_G(W)|$ for every subset $W\subseteq X$, then there exists a matching in $G_d$ by~\Cref{thm:hall:graphformulation}. It implies that a matching exists for each of the $d$ copies. These matchings correspond to $d$ disjoint matchings in $G$.  
\end{proof}

From the construction of our bipartite graph, we know that each element of $X$ has $\frac{(q+1)q}{2}\cdot 2 = q(q+1)$ edges because $|R_p|=q+1$ for all $p$. 
Each element of $Y$ has $(q^2-1)/(q-1) = (q+1)$ edges because a pair of indices occurs in exactly $(q^2-1)/(q-1)$ sets $R_p$ by \Cref{lem:countSubsetsContainPair}.  
Thus for any $W\subseteq X$, $|N_G(W)| \geq \frac{q(q+1)}{(q+1)}|W|=q|W|$. By~\Cref{cor:hall:dgraphformulation}, there exist $q$ disjoint matchings. This is equivalent to assigning $q$ non-central diagonal blocks to each processor.

To obtain $q$ disjoint matchings, we can use maximum cardinality matching algorithms, such as the Ford–Fulkerson algorithm~\cite{Ford87} and the Hopcroft-Karp algorithm~\cite{HK73}, on the transformed graph.

\paragraph{Data distribution of central diagonal blocks} 
There are $q^2+1 < P$ central diagonal blocks, so we want to assign 1 central diagonal block to some processors in such a way that no additional data from the vectors $\V{x} $ and $\V{y}$ is required for the processor to perform the computations associated with the diagonal block.
To demonstrate that there exists such an assignment, we again apply Hall's Marriage Theorem this time to the graph $G(X,Y,E)$ where $X=\{1,\ldots,q^2+1\}$, $Y=\{1,\ldots,P\}$ and $e=\{(x,y):x\in R_y\}$.
As $|R_p|=q+1$ for each $p$, therefore each element of $Y$ has $q+1$ edges. Each index appears in $q(q+1)$ sets $R_p$ by~\Cref{lem:countSubsetsContainElement}. Therefore, each element of $X$ has $q(q+1)$ edges. Thus for any $W\subseteq X$, $|N_G(W)| \geq \frac{q(q+1)}{(q+1)}|W|=q|W|$. By~\cref{thm:hall:graphformulation}, a matching can be found.
Again to find the matching we can use a maximum cardinality matching algorithm.

Recall that the block size is $b=n/(q^2+1)$ and $P=q(q^2+1)$. The numbers of lower tetrahedral elements in an off-diagonal block, a non-central diagonal block and a central diagonal block are $b^3$, $\frac{b^2(b+1)}{2}$ and $\frac{b(b+1)(b+2)}{6}$, respectively.
As each processor owns $(q+1)q(q-1)/6$ off-diagonal blocks, $q$ non-central diagonal blocks, and up to $1$ central diagonal block of the symmetric tensor, the processor stores at most
$$\frac{(q+1)q(q-1)}{6}b^3 + q\frac{b^2(b+1)}{2} + \frac{b(b+1)(b+2)}{6} \approx \frac{n^3}{6P}$$ 
elements of the tensor.

\section{Algorithm}
\label{sec:algo:3d}

Now we present our algorithm that attains the lower bound of \Cref{thm:memindeplb} matching the leading order constant exactly. It is outlined in \Cref{alg:sttvv}.
{\footnotesize
\begin{algorithm}
	\caption{Parallel \sttsv computation}
	\label{alg:sttvv}
	\begin{algorithmic}[1]
		\Require $q$ is a prime power
		\Require $b = n/(q^2+1)$ is the block size and $m=n/b$.
		\Require $\Pi$ is a set of processors with $|\Pi|=P=q(q^2+1)$.
		\Require $\T{A}[{T}_{p}] = \{ \T{A}[i][j][k]: i>j>k \in R_p \} \cup \{\T{A}[i][i][k]: (i,i,k)\in N_p\} \cup \{\T{A}[i][k][k]: (i,k,k)\in N_p\} \cup \{\T{A}[i][i][i]: (i,i,i)\in D_p\}$ is extended tetrahedral block owned by processor $p$.
		\Require $\V{x}$ is evenly divided into $m$ row blocks, and each row block $\V{x}[i]$ is evenly divided across a set of $q(q+1)$ processors $Q_i$ with $\V{x}[i]^{(p)}$ owned by processor $p$ so that $\V{x}[R_p]^{(p)} = \{\V{x}[i]^{(p)}: i \in R_p\}$ is a set of data blocks of $\V{x}$ owned by processor $p$. Similarly, $\V{y}[R_p]^{(p)}$ is defined and it is owned by processor  $p$.
		\Ensure $\V{y} = \T{A}\times_2 \V{x} \times_3 \V{x}$ with $\V{y}[R_p]^{(p)}$ owned by processor $p$.
        \Function{$\V{y}[R_p]^{(p)} = \text{\sttsv}$}{$\T{A}[{T}_{p}], \V{x}[R_p]^{(p)}, \Pi$}
		\State $p = \Call{MyRank}{\Pi}$			
		\Statex\Comment{Store its own data of a row in $q(q+1) -1 $ blocks}
		\State Allocate array $\bar{\V{x}}$ of $2P$ blocks, each of size $\frac{b}{q(q+1)}$ \label{line:sttvv:startcommx}
		\ForAll{$i \in R_p$}
		\ForAll{$p' \in Q_i \backslash \{p\}$}
		\State $\bar{\V{x}}[p'] = \V{x}[i]^{(p)}$
		\EndFor
		\EndFor
		\State $\bar{\V{x}}$ = \Call{All-to-All}{$\bar{\V{x}},\Pi$}\label{line:sttvv:AlltoAllx}
		\ForAll{$i \in R_p$}
		\ForAll{$p' \in Q_i \backslash \{p\}$}
		\State Store $\bar{\V{x}}[p']$ into $\V{x}[i]$
		\EndFor
		\EndFor \label{line:sttvv:stopcommx}
        \Statex\Comment{Perform local computations (ternary multiplications)}
		\State Initialize portions of row blocks of $\V{y}$ own by $p$ to $0$
		\ForAll{$\T{A}[i][j][k] \in \T{A}[{T}_{k}]$} \label{line:sttvv:startcomp}
		\If{$i\neq j \neq k$}
		\State $\V{y}[i] = 2\T{A}[i][j][k] \times_2 \V{x}[j] \times_3 \V{x}[k]$
		\State $\V{y}[j] = 2\T{A}[i][j][k] \times_1 \V{x}[i] \times_3 \V{x}[k]$ 
		\State $\V{y}[k] = 2\T{A}[i][j][k] \times_1 \V{x}[i] \times_2 \V{x}[j]$  
		\ElsIf{$i==j$ and $j\neq k$}
		\State $\V{y}[i] = 2\T{A}[i][i][k] \times_2 \V{x}[i] \times_3 \V{x}[k]$; $\ \V{y}[k] = \T{A}[i][i][k] \times_1 \V{x}[i] \times_2 \V{x}[i]$
		\ElsIf{$i\neq j$ and $j == k$}
		\State $\V{y}[i] = \T{A}[i][k][k] \times_2 \V{x}[k] \times_3 \V{x}[k]$; $\ \V{y}[k] = 2\T{A}[i][k][k] \times_1 \V{x}[i] \times_2 \V{x}[k]$
		\Else
		\State $\V{y}[i] = \T{A}[i][i][i] \times_2 \V{x}[i] \times_3 \V{x}[i]$
		\EndIf
		\EndFor \label{line:sttvv:stopcomp}
		\Statex\Comment{Distribute and reduce results for each row block set}
		\State Allocate array $\bar{\V{y}}$ of $2P$ blocks, each of size $\frac{b}{q(q+1)}$\label{line:sttvv:startcommy}
		\ForAll{$i \in R_p$}
		\ForAll{$p' \in Q_i \backslash \{p\}$}
		\State $\bar{\V{y}}[p'] = \V{y}[i]^{(p')}$
		\EndFor
		\EndFor
		\State $\bar{\V{y}}$ = \Call{All-to-All}{$\bar{\V{y}},\Pi$}\label{line:sttvv:AlltoAlly}
		\ForAll{$i \in R_p$}
		\State $\V{y}[i]^{(p)} \pluseq \tilde{\V{y}}[i]^{(p)}$
		\ForAll{$p' \in Q_i \backslash \{p\}$}
		\State $\V{y}[i]^{(p)} \pluseq \bar{\V{y}}[p']$\label{line:sttvv:reduceyi}
		\EndFor
		\EndFor\label{line:sttvv:stopcommy}
		\EndFunction
	\end{algorithmic}
\end{algorithm}
}

We use $\Pi$ to denote the set of all processors where $|\Pi|=P=q(q^2+1)$ for a prime power $q$.
We use $\V{x}[i]$ to denote the $i$th row block of $\V{x}$, and $\T{A}[i][j][k]$ to denote the block of $\T{A}$ at index $(i,j,k)$.
To perform the local computation, each processor $p$ must gather $q+1$ row blocks corresponding to the set $R_p$ for the vector $\V{x}$. 
The processor owns only the $\V{x}[i]^{(p)}$ portion of the $\V{x}[i]$ row block for each $i\in R_p$, therefore it must communicate with the other processors in $Q_i$ to gather the full $\V{x}[i]$ block. 
As any processor may communicate to at most $P-1$ processors, we expressed our communication pattern with a \Call{All-to-All}{} collective which is presented in \cref{line:sttvv:startcommx} to \cref{line:sttvv:stopcommx}. 
Most of the pseudocode involves packing and unpacking of a temporary array $\bar{\V{x}}$ to organize the data for communication. 
Each processor performs its share of local computations (ternary multiplications) in  \cref{line:sttvv:startcomp} to \cref{line:sttvv:stopcomp} to compute its partial contributions to $\V{y}$. 
Then, each processor $p$ must share the partially computed $\V{y}[i]$ with the other processors in $Q_i$ to obtain the final $\V{y}[i]^{(p)}$ portion of $\V{y}[i]$ for each $i\in R_p$. 
This communication pattern is also expressed with an \Call{All-to-All}{} collective, which is presented in \cref{line:sttvv:startcommy} to \cref{line:sttvv:stopcommy}. 
Similar to the communication of $\V{x}$, a temporary array $\bar{\V{y}}$ is employed to pack and unpack the data for communication. 
The processor $p$ sums all partial results in \cref{line:sttvv:reduceyi} to obtain the final $\V{y}[i]^{(p)}$. 

\subsection{Computational Cost Analysis}
\label{sec:algo:compcost}

We express the computational cost of the algorithm in terms of number of ternary multiplications performed. 
Recall that the dimension of each row block of a vector is $b=n/(q^2+1)$, and the dimensions of each block of the symmetric tensor are $b \times b\times b$. 
We analyze the amount of computations performed for each type of block of the symmetric tensor. 
There are $b^3$ strict lower tetrahedral elements of $\T{A}$ in an off-diagonal block, therefore the number of ternary multiplications performed by the algorithm for each block is $3b^3$. 
There are $b^2(b+1)/2$ lower tetrahedral elements of $\T{A}$ in a non-central diagonal block, of which $b^2(b-1)/2$ are strictly lower tetrahedral elements and $b^2$ are non-central diagonal elements. 
Therefore, the number of ternary multiplications performed by the algorithm for each non-central diagonal block is  $3b^2(b-1)/2 + 2b^2$. 
A central diagonal block contains $b(b+1)(b+2)/6$ lower tetrahedral elements of $\T{A}$, among which $b(b-1)(b-2)/6$ are strictly lower tetrahedral elements, $b(b-1)$ are non-central diagonal elements and $b$ are central diagonal elements. 
Therefore, the number of ternary multiplications performed by the algorithm for each central diagonal block is $3b(b-1)(b-2)/6 + 2b(b-1) + b$.

All ternary multiplications occur between \cref{line:sttvv:startcomp} and \cref{line:sttvv:stopcomp} in \Cref{alg:sttvv}. 
For every $p$, we have $|R_p|=q+1$, $|N_p| = q$ and $|D_p| \in \{0,1\}$. 
Thus the computational cost of the algorithm, in terms of ternary multiplications performed, is at most 
$$\frac{(q+1)q(q-1)}{6}\cdot 3b^3 + q \cdot 3b^2(b-1) + \frac{3b(b-1)(b-2)}{6}\text .$$
As $q(q^2+1)=P$ and $b=n/(q^2+1)$, the leading order terms of the computational cost is
$$ \frac{n^3}{2q^3} + \frac{3n^3}{q^5} + \frac{n^3}{2q^6} = \frac{n^3}{2P} + o\left(\frac{n^3}{P}\right)\text.$$

Note that we may not have perfect computational load balance, as not all processors are assigned a central diagonal block. 
However the imbalance does not affect the constant in the leading order term.

\subsection{Communication Cost Analysis}
\label{sec:algo:commcost}

\subsubsection{New Results on Communication Costs}
\label{sec:algo:stepscomm}
Here we present results that we will use while analyzing communication cost of our parallel algorithm in \Cref{sec:algo:3d}. 
One can easily obtain the following lemma by applying \Cref{thm:hall:graphformulation} on a bipartite graph where each vertex in $X$ and $Y$ is repeated $d$ times.

\begin{lemma}
	\label{lem:balanced:bipartitegraph}
	Let $G(X,Y)$ be a finite bipartite graph with bipartite sets $X$ and $Y$, where $|X|=|Y|$, each vertex in $X$ is connected to exactly $d$ vertices in $Y$, and vice versa. 
	There exist $d$ disjoint matchings in $G$, and each of them covers every vertex in $X$ and $Y$.
\end{lemma}

We can use maximum cardinality matching algorithms, such as the Ford–Fulkerson algorithm~\cite{Ford87} and the Hopcroft-Karp algorithm~\cite{HK73}, on $G$ to obtain $d$ such matchings.

We will use \Cref{lem:balanced:bipartitegraph} to bound the number of words sent by any processor in a set of processors, where each processor communicates with the same number of processors. 
We remind the reader that, in our model, a processor can send and receive at most one message at the same time.

\begin{theorem}
	\label{thm:proc:boundWords}
	Consider a set of $P$ processors. 
	Every processor needs to send and receive one message of $W$ words to and from each processor in a set of $d\; (d\leq P-1)$ processors. 
	Then, all required communication can be performed in $d$ steps, with the maximum number of words sent or received by any processor being $dW$. 
\end{theorem}
\begin{proof}
	Consider the bipartite graph $G(X,Y)$ with $X=Y=\{1,\ldots,P\}$ corresponding to processors. 
	If a processor $i$ needs to send a message to processor $j$, then there is an edge from $i\in X$ to $j\in Y$. 
	From the construction, each vertex in $X$ is connected to exactly $d$ vertices in $Y$, and vice versa. 
	By~\Cref{lem:balanced:bipartitegraph}, there exist $d$ matchings and each of them covers all the vertices of $X$ and $Y$. 
	We can see that each matching corresponds to a communication step that describes how processors can communicate. 
	Every processor sends and receives $W$ words in each step. 
	As there are $d$ different matchings, every processor can perform all required communication in $d$ steps and the total number of words sent or received by the processor is $dW$.	
\end{proof}

\subsubsection{Communication Cost of \Cref{alg:sttvv}}
The communication is performed in \cref{line:sttvv:AlltoAllx} to collect all required row blocks of the input vector $\V{x}$ and in \cref{line:sttvv:AlltoAlly} to distribute the partially computed row blocks of the output vector $\V{y}$. 

Each row block of a vector is shared evenly among $q(q+1)$ processors. 
Therefore, each participating processor must send $\frac{b}{q(q+1)}$ words of the block to all other participating processors and receive the same number of words from each of them. 
For every $p$, $|R_p|=q+1$, therefore each processor owns portions of $q+1$ row blocks of a vector. 
Thus, each processor must send  $(q+1)  \frac{b}{q(q+1)}\left(q(q+1) -1\right) = n\frac{q+1}{q^2+1} -\frac{n}{P}$ words and receive the same number of words, in total, for each vector.

In~\Cref{tab:TBP-10-30}, we can observe that processor $1$ does not share any data with processor $26$. 
This indicates that each processor does not need to communicate to all other processors. 
We expressed communication in our algorithm using \Call{All-to-All}{} collectives to simplify our presentation. 
However, it can be accomplished by directly communicating between source and destination. 
We will show that it is possible to perform only the required amount of data transfers for each vector in $(q+1)\left(q(q+1) -1\right) - q\frac{q(q+1)}{2}= \frac{q^3}{2} + \frac{3q^2}{2} -1$ steps.

As each processor owns row blocks of a vector defined by a tetrahedral block partition, any pair of processors can share at most $2$ row blocks. 
Any pair of row blocks is present in $q+1$ sets $R_p$ by \Cref{lem:countSubsetsContainPair}, each participating processor will communicate data its data from $2$ row blocks with other participating processors. 
There are $q(q+1)/2$ possible pairs for a $R_p$ set, therefore each processor will communicate $2$ row blocks with $\frac{q(q+1)}{2}(q+1-1) =\frac{q^2(q+1)}{2}$ processors. 
By~\Cref{thm:proc:boundWords}, all required communication can be performed in $\frac{q^2(q+1)}{2}$ steps, where each processor sends and receives $\frac{q^2(q+1)}{2} \cdot \frac{2b}{q(q+1)}$ words. 
The distribution presented in \Cref{tab:TBP-10-30} shows that each processor will communicate $2$ blocks with $18$ processor for each vector.

Now we determine the number of processors with which a processor will communicate only one row block of each vector.
Our algorithm requires each processor communicates $(q(q+1)-1)(q+1)$ row blocks, where $2$ blocks must be sent to each of $\frac{q^2(q+1)}{2}$ processors. 
Therefore, each processor will communicate only one row block with  $q^2-1$ processors. By~\Cref{thm:proc:boundWords}, all required communications can be performed in $q^2-1$ steps, where each processor sends and receives $(q^2-1) \cdot \frac{b}{q(q+1)}$ words.
In our working example processor 1 communicates only one row block with processors $8, 11, 16, 19, 21, 24, 27$ and $30$.

Thus each processor sends and receives $q^2(q+1) \cdot \frac{b}{q(q+1)} + (q^2-1) \cdot \frac{b}{q(q+1)} = n\frac{q+1}{q^2+1} -\frac{n}{P}$ words in $\frac{q^2(q+1)}{2}+q^2-1 = \frac{q^3}{2} + \frac{3q^2}{2} -1$ steps. We show a sequence of communication steps for a relatively small example in \Cref{sec:anotherExample}.

As communication is performed for both vectors, the bandwidth cost of the algorithm is
$$2\left(n\frac{q+1}{q^2+1} -\frac{n}{P}\right)\text.$$

As $P=q(q^2+1)$, which implies that $(q^2+1)/(q+1) \approx P^{1/3}$, the bandwidth cost is dominated by $ 2(n/P^{1/3} - n/P)$. Hence the bandwidth cost of the algorithm matches the leading term in the lower bound exactly.

\paragraph{Communication cost of our algorithm with All-to-All collectives}
Here we discuss bandwidth cost of our algorithm when \Call{All-to-All}{} collectives are used to communicate both vectors $\V{x}$ and $\V{y}$. An \Call{All-to-All}{} collective, using an algorithm that achieves optimal bandwidth cost, takes $P-1$ steps~\cite{Thakur:CollectiveCommunications:2005}.
For each vector, at each step, a processor may send its own data of $2$ row blocks to another processor. 
Therefore, the bandwidth cost of each step is $2\frac{b}{q(q+1)}$, and the overall bandwidth cost of an \Call{All-to-All}{} collective is $2\frac{b}{q(q+1)} (P-1) = \frac{2n}{q+1} (1-1/P)$. 
As communication is performed for both vectors, the bandwidth cost of the algorithm using \Call{All-to-All}{} collectives is $\frac{4n}{q+1} (1-1/P)$. 
Given that $P=q(q^2+1)$, which implies $q\approx P^{1/3}$, the bandwidth cost is dominated by $4n/P^{1/3}$, which is twice the leading term in the lower bound.

\section{Conclusion}
\label{sec:conclusion}
In this work, we establish communication lower bounds for the parallel \sttsv computation. The lower bound on accessed data is formulated as the solution of a constrained optimization problem. One of the constraints is obtained by manipulating the symmetric iteration space and deriving a relation between the number of iteration points and the required data access within vectors. We also show that our bounds are tight by presenting a communication optimal algorithm. We extend the triangle block partition scheme proposed by Beaumont et al.~\cite{BELV22} and Al Daas et al.~\cite{ABGKR23} for the lower triangle of symmetric matrices to a tetrahedral partition of the lower tetrahedron of 3-dimensional symmetric tensors. Our algorithm works with tetrahedral blocks that reduces data accesses within vectors. We use Steiner systems, which provide collections of equally sized subsets of indices where any three distinct indices appear in exactly one subset, to generate the tetrahedral partitions.

We focus on \Cref{alg:sttvv:symmetry} that performs $n^2(n+1)/2$ ternary multiplications. Each \arraymul requires $2$ elementary multiplications. Furthermore, an addition and often a multiplication are involved with each \arraymul. Therefore, the algorithm needs to perform a total of approximately $2n^3$ elementary arithmetic operations.  In our approach, we assume that each \arraymul is computed on a single processor, i.e, partial ternary multiplications are not reused across processors. However it is reasonable for an algorithm to break this assumption to improve arithmetic costs by reusing partial results across processors. For example, \sttsv can also be performed in sequence, first $\T{A}\times_2x$ is computed with parallel matrix multiplication and then the result is multiplied with $x$ in parallel to obtain the final output. This approach requires performing a total of $2n^3 + 2n^2$ elementary arithmetic operations. However, by leveraging symmetry, the total number of operations can be reduced by almost half. When $P$ is small ($P\leq n)$, bandwidth cost of the first step is at least $\mathcal{O}(n)$~\cite{ABGKR22}, which is asymptotically larger than that of our algorithm. We plan to study the sequence approach for parallel \sttsv computation in detail in the future.

We intend to generalize our results for $d$-dimensional computations. The lower bound arguments can easily be extended for $d$-dimensional \sttsv computations. However, there are no known infinite families of Steiner systems for $s>3$, so constructing partitions for higher dimensions is more challenging.

In our computational model, we assume that  each processor has sufficiently large local memory. Therefore, our algorithm may not be feasible in limited-memory scenarios. We plan to explore parallel algorithms that are communication optimal for such scenarios. 

Matricized tensor times Khatri-Rao product (MTTKRP) is a bottleneck operation for algorithms to compute a Canonical Polyadic tensor decomposition. Mode-1 MTTKRP for a $3$-dimensional symmetric tensor $\T{A}$ can be computed element-wise as 
$$Y_{i\ell} = \sum_{j,k}\T{A}_{ijk}\cdot X_{j\ell} \cdot X_{k\ell}\text.$$ 

Here $X$ is the factor matrix associated with the decomposition. For a fixed $\ell$, the above expression represents a \sttsv computation.
We plan to generalize our lower bound and algorithm to MTTKRP in the future.

\begin{acks}
This project has received funding through the UKRI Digital Research Infrastructure Programme through the Science and Technology Facilities Council’s
Computational Science Centre for Research Communities (CoSeC).
    This project received funding from the European Research Council (ERC) under European Union's Horizon 2020 research and innovation program (grant agreement 810367).  
    This work is supported by the National Science Foundation under grant CCF-1942892. This material is based upon work supported by the US Department of Energy, Office of Science, Advanced Scientific Computing Research program under award DE-SC-0023296.
\end{acks}

\bibliographystyle{ACM-Reference-Format}
\bibliography{refs}

\appendix
\newpage
\section{Another Example of Tetrahedral Block Partition}
\label{sec:anotherExample}

While our main text focused on Steiner $(n,r,s)$ systems in the infinite family from finite spherical geometries, there are many more Steiner $(n,r,3)$ systems which can be used to generate tetrahedral block partitions.
We will use a smaller system to demonstrate a sequence of steps for the required point-to-point communication.
\Cref{app:tab:TBP-PRS-8-14} shows the assignment of all lower tetrahedron blocks of a symmetric tensor from a Steiner $(8,4,3)$ system. 
This system does not belong to the family of Steiner $(q^\alpha +1, q+1, 3)$ systems, but exists in the literature.

\begin{table}[ht]
	\begin{tabular}{|c|c|c|c||c|c|}
		\hline
		$p$ & $R_p$ & $N_p$ & $D_p$ & $i$ & $Q_i$\\
		\hline
		1 & \{1,2,3,4\} & \{(2,2,1), (3,3,2), (2,1,1), (3,2,2)\} & \{(1,1,1)\} & 1 & \{1,2,3,4,5,6,7\}\\
		2 & \{1,2,5,6\} & \{(5,5,1), (6,6,1), (5,1,1), (5,2,2)\} & \{(2,2,2)\} & 2 & \{1,2,3,8,9,10,11\}\\
		3 & \{1,2,7,8\} & \{(7,7,1), (8,8,1), (7,1,1), (7,2,2)\} & \{(7,7,7)\} & 3 & \{1,4,5,8,9,12,13\}\\
		4 & \{1,3,5,7\} & \{(7,7,3), (7,7,5), (3,1,1), (7,3,3)\} & \{(3,3,3)\} & 4 & \{1,6,7,10,11,12,13\}\\
		5 & \{1,3,6,8\} & \{(6,6,3), (3,3,1), (6,1,1), (8,1,1)\} & \{(6,6,6)\} & 5 & \{2,4,6,8,10,12,14\}\\
		6 & \{1,4,5,8\} & \{(8,8,4), (5,5,4), (4,1,1), (5,4,4)\} & \{(5,5,5)\} & 6 & \{2,5,7,9,11,12,14\}\\
		7 & \{1,4,6,7\} & \{(7,7,4), (4,4,1), (6,4,4), (7,6,6)\} & \{(4,4,4)\} & 7 & \{3,4,7,9,10,13,14\}\\
		8 & \{2,3,5,8\} & \{(8,8,5), (5,5,3), (5,3,3), (8,2,2)\} & \{(8,8,8)\} & 8 & \{3,5,6,8,11,13,14\}\\
		9 & \{2,3,6,7\} & \{(6,6,2), (7,7,2), (6,2,2), (6,3,3)\} & \{\} & & \\
		10 & \{2,4,5,7\} & \{(5,5,2), (4,4,2), (4,2,2), (7,4,4)\} & \{\} & & \\
		11 & \{2,4,6,8\} & \{(8,8,2), (8,8,6), (8,4,4), (8,6,6)\} & \{\} & & \\
		12 & \{3,4,5,6\} & \{(6,6,4), (4,4,3), (4,3,3), (6,5,5)\} & \{\} & & \\
		13 & \{3,4,7,8\} & \{(8,8,3), (8,8,7), (8,3,3), (8,7,7)\} & \{\} & & \\
		14 & \{5,6,7,8\} & \{(6,6,5), (7,7,6), (7,5,5), (8,5,5)\} & \{\} & & \\
		\hline
	\end{tabular}
	\caption{Processor and row block sets of tetrahedral block partition for $m=8$ and $P=14$. Processor $p$ has all $TB_3(R_p)\cup N_p\cup D_p$ blocks of the symmetric tensor and performs all required computations with them. Row block $i$ of a vector is evenly distributed  among all processors of $Q_i$.} \label{app:tab:TBP-PRS-8-14}
\end{table}

\Cref{fig:communicationsteps} shows a sequence of steps to perform all required communication for the assignment of \Cref{app:tab:TBP-PRS-8-14} with $m=8$ and $P=14$. In each step, every processor directly communicates with another processor. Note that the number of steps is $12$, which is less than $P-1$.

\begin{figure}[htb]
	\begin{center}
		\subfloat[{\scriptsize Step 1.}]{
			\def\seq{1/1,2/13,3/14,4/12,5/11,6/10,7/8,8/9,9/7,10/6,11/5,12/4,13/3,14/2}
			\centering

\pgfmathsetmacro{\offset}{.33}

\begin{tikzpicture}[scale=0.45, every node/.style={transform shape}]
	\foreach \x/\y in \seq {
		\node[circle, draw, fill=blue!15, minimum size=5mm] (N\x) at (\x-0.5, 0) {\y};
	}
	\foreach \i in {1,...,13} {
		\pgfmathtruncatemacro{\next}{\i+1}
		\draw[->] (N\i) -- (N\next);
	}
	\draw[->] (N14) to [out=150,in=35] (N1);	
\end{tikzpicture}}$\qquad$		
		\subfloat[{\scriptsize Step 2.}]{
			\def\seq{1/1,2/13,3/14,4/12,5/11,6/10,7/8,8/9,9/7,10/6,11/5,12/4,13/3,14/2}
			\centering

\pgfmathsetmacro{\offset}{.33}

\begin{tikzpicture}[scale=0.45, every node/.style={transform shape}]
	\foreach \x/\y in \seq {
		\node[circle, draw, fill=blue!15, minimum size=5mm] (N\x) at (\x-0.5, 0) {\y};
	}
	\foreach \i in {1,...,13} {
		\pgfmathtruncatemacro{\next}{\i+1}
		\draw[<-] (N\i) -- (N\next);
	}
	\draw[<-] (N14) to [out=150,in=35] (N1);	
\end{tikzpicture}}
		
		\subfloat[{\scriptsize Step 3.}]{
			\def\seq{1/1,2/12,3/13,4/9,5/14,6/10,7/7,8/11,9/8,10/6,11/4,12/2,13/5,14/3}
			\centering

\pgfmathsetmacro{\offset}{.33}

\begin{tikzpicture}[scale=0.45, every node/.style={transform shape}]
	\foreach \x/\y in \seq {
		\node[circle, draw, fill=blue!15, minimum size=5mm] (N\x) at (\x-0.5, 0) {\y};
	}
	\foreach \i in {1,...,13} {
		\pgfmathtruncatemacro{\next}{\i+1}
		\draw[->] (N\i) -- (N\next);
	}
	\draw[->] (N14) to [out=150,in=35] (N1);	
\end{tikzpicture}}$\qquad$
		\subfloat[{\scriptsize Step 4.}]{
			\def\seq{1/1,2/12,3/13,4/9,5/14,6/10,7/7,8/11,9/8,10/6,11/4,12/2,13/5,14/3}
			\centering

\pgfmathsetmacro{\offset}{.33}

\begin{tikzpicture}[scale=0.45, every node/.style={transform shape}]
	\foreach \x/\y in \seq {
		\node[circle, draw, fill=blue!15, minimum size=5mm] (N\x) at (\x-0.5, 0) {\y};
	}
	\foreach \i in {1,...,13} {
		\pgfmathtruncatemacro{\next}{\i+1}
		\draw[<-] (N\i) -- (N\next);
	}
	\draw[<-] (N14) to [out=150,in=35] (N1);	
\end{tikzpicture}}
		
		\subfloat[{\scriptsize Step 5.}]{
			\def\seq{1/1,2/8,3/13,4/6,5/10,6/9,7/5,8/12,9/4,10/7,11/3,12/14,13/2,14/11}
			\centering

\pgfmathsetmacro{\offset}{.33}

\begin{tikzpicture}[scale=0.45, every node/.style={transform shape}]
	\foreach \x/\y in \seq {
		\node[circle, draw, fill=blue!15, minimum size=5mm] (N\x) at (\x-0.5, 0) {\y};
	}
	\foreach \i in {1,...,13} {
		\pgfmathtruncatemacro{\next}{\i+1}
		\draw[->] (N\i) -- (N\next);
	}
	\draw[->] (N14) to [out=150,in=35] (N1);	
\end{tikzpicture}}$\qquad$
		\subfloat[{\scriptsize Step 6.}]{
			\def\seq{1/1,2/8,3/13,4/6,5/10,6/9,7/5,8/12,9/4,10/7,11/3,12/14,13/2,14/11}
			\centering

\pgfmathsetmacro{\offset}{.33}

\begin{tikzpicture}[scale=0.45, every node/.style={transform shape}]
	\foreach \x/\y in \seq {
		\node[circle, draw, fill=blue!15, minimum size=5mm] (N\x) at (\x-0.5, 0) {\y};
	}
	\foreach \i in {1,...,13} {
		\pgfmathtruncatemacro{\next}{\i+1}
		\draw[<-] (N\i) -- (N\next);
	}
	\draw[<-] (N14) to [out=150,in=35] (N1);	
\end{tikzpicture}}
		
		\subfloat[{\scriptsize Step 7.}]{
			\def\seq{1/1,2/4,3/8,4/3,5/6,6/14,7/7,8/5,9/13,10/11,11/9,12/2,13/12,14/10}
			\centering

\pgfmathsetmacro{\offset}{.33}

\begin{tikzpicture}[scale=0.45, every node/.style={transform shape}]
	\foreach \x/\y in \seq {
		\node[circle, draw, fill=blue!15, minimum size=5mm] (N\x) at (\x-0.5, 0) {\y};
	}
	\foreach \i in {1,...,13} {
		\pgfmathtruncatemacro{\next}{\i+1}
		\draw[->] (N\i) -- (N\next);
	}
	\draw[->] (N14) to [out=150,in=35] (N1);	
\end{tikzpicture}}$\qquad$
		\subfloat[{\scriptsize Step 8.}]{
			\def\seq{1/1,2/4,3/8,4/3,5/6,6/14,7/7,8/5,9/13,10/11,11/9,12/2,13/12,14/10}
			\centering

\pgfmathsetmacro{\offset}{.33}

\begin{tikzpicture}[scale=0.45, every node/.style={transform shape}]
	\foreach \x/\y in \seq {
		\node[circle, draw, fill=blue!15, minimum size=5mm] (N\x) at (\x-0.5, 0) {\y};
	}
	\foreach \i in {1,...,13} {
		\pgfmathtruncatemacro{\next}{\i+1}
		\draw[<-] (N\i) -- (N\next);
	}
	\draw[<-] (N14) to [out=150,in=35] (N1);	
\end{tikzpicture}}			
		
		\subfloat[{\scriptsize Step 9.}]{
			\def\seq{1/1,2/7,3/12,4/9,5/4,6/14,7/5,8/8,9/2,10/10,11/13,12/3,13/11,14/6}
			\centering

\pgfmathsetmacro{\offset}{.33}

\begin{tikzpicture}[scale=0.45, every node/.style={transform shape}]
	\foreach \x/\y in \seq {
		\node[circle, draw, fill=blue!15, minimum size=5mm] (N\x) at (\x-0.5, 0) {\y};
	}
	\foreach \i in {1,...,13} {
		\pgfmathtruncatemacro{\next}{\i+1}
		\draw[->] (N\i) -- (N\next);
	}
	\draw[->] (N14) to [out=150,in=35] (N1);	
\end{tikzpicture}}$\qquad$
		\subfloat[{\scriptsize Step 10.}]{
			\def\seq{1/1,2/7,3/12,4/9,5/4,6/14,7/5,8/8,9/2,10/10,11/13,12/3,13/11,14/6}
			\centering

\pgfmathsetmacro{\offset}{.33}

\begin{tikzpicture}[scale=0.45, every node/.style={transform shape}]
	\foreach \x/\y in \seq {
		\node[circle, draw, fill=blue!15, minimum size=5mm] (N\x) at (\x-0.5, 0) {\y};
	}
	\foreach \i in {1,...,13} {
		\pgfmathtruncatemacro{\next}{\i+1}
		\draw[<-] (N\i) -- (N\next);
	}
	\draw[<-] (N14) to [out=150,in=35] (N1);	
\end{tikzpicture}}
		
		\subfloat[{\scriptsize Step 11.}]{
			\def\seq{1/1,2/9,3/3,4/10,5/4,6/13,7/7,8/2,9/6,10/12,11/8,12/14,13/11,14/5}
			\centering

\pgfmathsetmacro{\offset}{.33}

\begin{tikzpicture}[scale=0.45, every node/.style={transform shape}]
	\foreach \x/\y in \seq {
		\node[circle, draw, fill=blue!15, minimum size=5mm] (N\x) at (\x-0.5, 0) {\y};
	}
	\foreach \i in {1,...,13} {
		\pgfmathtruncatemacro{\next}{\i+1}
		\draw[->] (N\i) -- (N\next);
	}
	\draw[->] (N14) to [out=150,in=35] (N1);	
\end{tikzpicture}}$\qquad$
		\subfloat[{\scriptsize Step 12.}]{
			\def\seq{1/1,2/9,3/3,4/10,5/4,6/13,7/7,8/2,9/6,10/12,11/8,12/14,13/11,14/5}
			\centering

\pgfmathsetmacro{\offset}{.33}

\begin{tikzpicture}[scale=0.45, every node/.style={transform shape}]
	\foreach \x/\y in \seq {
		\node[circle, draw, fill=blue!15, minimum size=5mm] (N\x) at (\x-0.5, 0) {\y};
	}
	\foreach \i in {1,...,13} {
		\pgfmathtruncatemacro{\next}{\i+1}
		\draw[<-] (N\i) -- (N\next);
	}
	\draw[<-] (N14) to [out=150,in=35] (N1);	
\end{tikzpicture}}
	\end{center}
	\caption{A sequence of $12$ communication steps to perform all required data transfers among all $14$ processors for the tetrahedral block partition of \Cref{app:tab:TBP-PRS-8-14}. $i\rightarrow j$ indicates that processor $i$ sends a message to processor $j$. In each step, a processor sends and receives exactly one message.\label{fig:communicationsteps}}
\end{figure}
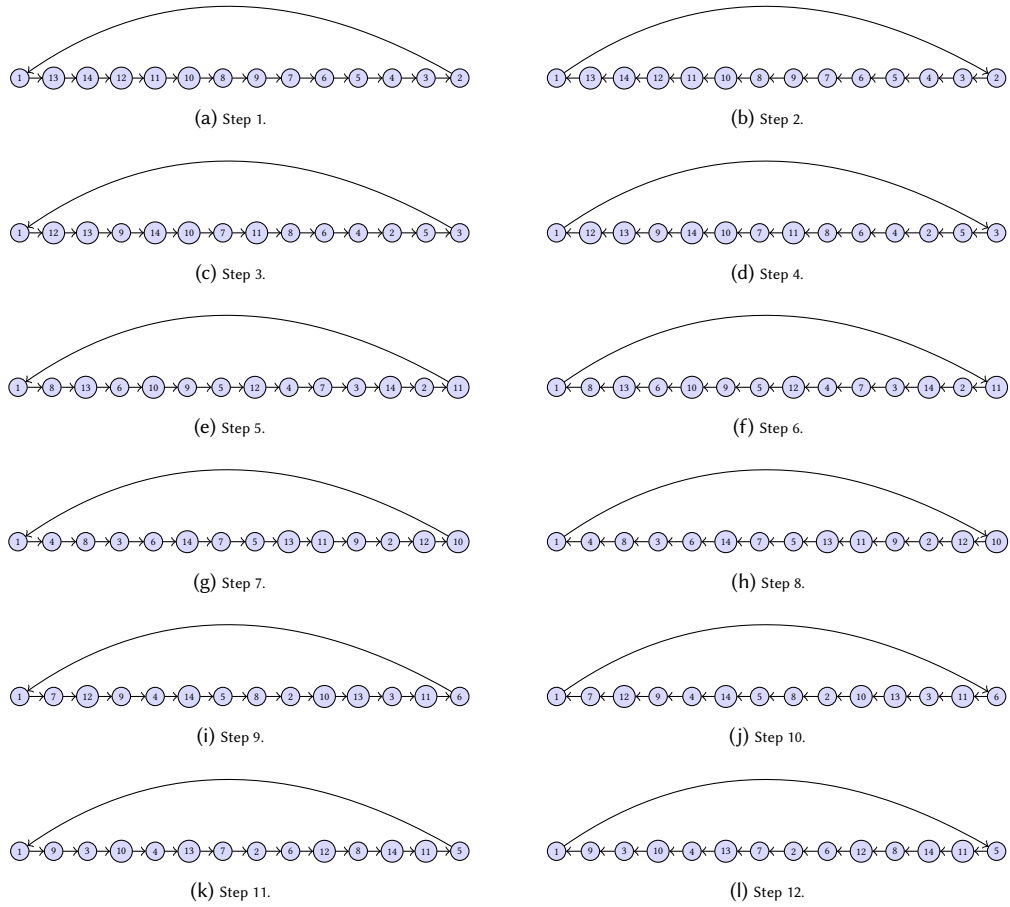

\end{document}